\documentclass[journal]{IEEEtran}
\usepackage{amssymb,amsmath,amsthm}
\usepackage{url,color,lipsum}
\usepackage{graphicx}

\newtheorem{A}{Assumption}
\newtheorem{D}{Definition}
\newtheorem{T}{Theorem}
\newtheorem{Prop}{Proposition}

\ifCLASSINFOpdf

\else

\fi

\begin{document}

\title{Motion Optimization for Musculoskeletal Dynamics:\\ A Flatness-Based Polynomial Approach}

\author{Hanz Richter,~\IEEEmembership{Member,~IEEE,}
        and~Holly Warner
\thanks{Authors are with the Department
of Mechanical Engineering, Cleveland State University, Cleveland,
OH, 44115 USA e-mail: h.richter@csuohio.edu}
}


\maketitle

\begin{abstract}
A new approach for trajectory optimization of musculoskeletal dynamic models is introduced. The model combines rigid body and muscle dynamics described with a Hill-type model driven by neural control inputs. The objective is to find input and state trajectories which are optimal with respect to a minimum-effort objective and meet constraints associated with musculoskeletal models. The measure of effort is given by the integral of pairwise average forces of the agonist-antagonist muscles. The concepts of flat parameterization of nonlinear systems and sum-of-squares optimization are combined to yield a method that eliminates the numerous set of dynamic constraints present in collocation methods. With terminal equilibrium, optimization reduces to a feasible linear program, and a recursive feasibility proof is given for more general polynomial optimization cases. The methods of the paper can be used as a basis for fast and efficient solvers for hierarchical and receding-horizon control schemes. Two simulation examples are included to illustrate the proposed methods.        
\end{abstract}

\IEEEpeerreviewmaketitle

\section{Introduction}\label{intro}
\IEEEPARstart{M}{usculoskeletal} system (MSS) models are used to describe the dynamics of human movement. MSS combine rigid-body models with nonlinear dynamic descriptions of muscle force production with neural stimulus signals as control inputs~\cite{geyer2019neuromuscular,PandySHBED,nigg07}. MSS models form the analytical basis for studies aimed at the identification of human control. A postulated human control system must not only support prediction accuracy, but also through its properties reflect key traits of human movement. Specifically, humans can reach and regulate positions or track trajectories if so demanded. Also, humans maintain stability and can complete their tasks under unknown loads. Finally, the MSS is an over-actuated system and incorporates a redundancy resolution mechanism that has been widely regarded to be a form of optimal control~\cite{acker10} based on the minimization of effort under constraints. An account of the optimization objective functions that have been used appears in~\cite{erdemir07}. 

The specification of prospective control systems with the above characteristics is a challenging task, because MSS are large-scale, nonlinear over-actuated systems with constraints in the states and control inputs.  By and large, the literature is concerned with finite-horizon, open-loop optimal controls~\cite{todorov02}. Solution methods based on collocation \cite{witkin1988spacetime} are widely used, whereby candidate solutions are discretized into a large number of temporal nodes~\cite{betts1998survey,betts2010practical,kaplan2001predictive}. The vector of states and control values at such nodes constitutes the search variable and is used to evaluate the cost function and problem-related constraints. Dynamic constraints are accounted for by using a finite-differences approximation to the state derivative and introducing pertinent equality constraints, one for each pair of successive nodes.  The result of this formulation is a large-scale nonlinear static optimization problem. For instance, the model used in Section~\ref{simex} has 16 state variables and 6 control inputs. A direct collocation method with 100 temporal nodes and backward Euler discretization would result in 2200 search variables and 1584 constraints associated with the dynamics alone.

In contrast to finite-horizon optimal controls, human motion strategies may not involve a known, finite duration. Human control laws that maintain their properties for an indefinite period of time are thus required. Receding-horizon (RH) approaches such as Model Predictive Control (MPC) provide a method to use open-loop solutions indefinitely, establishing a feedback process~\cite{Mehrabi17}. Feedback introduces the capacity for adjustment to unanticipated changes in reference trajectories and external disturbances.

This paper introduces a flatness parameterization of the MSS models having a meaningful biomechanical interpretation. Moreover, a two-stage optimal control problem is formulated using the tools of sum-of-squares optimization and semidefinite programming. This approach removes all dynamic constraints from the optimization, significantly reducing the size of the problem. Further, two theoretical results are given to support the use of the approach as a fast, efficient solver to be used in conjunction with feedback implementations of optimal control such as MPC. Specifically, recursive feasibility of the second-stage optimization and its reduction to a linear program are shown.
 \subsection{Overview of Flatness Parameterizations}\label{flat}
Differential flatness is a property of dynamical systems that was extensively studied by Fliess and co-workers~\cite{FliessIJC95} in the early 1990s. Flatness implies that it is possible to parameterize each state and input of the system in terms of a set of variables known as flat outputs, without integration. When the system has $m$ inputs, exactly $m$ flat outputs are required.  Once a suitable set of flat outputs has been defined, congruent system trajectories and control inputs can be generated by direct evaluation, i.e., without having to solve differential equations. 

Physical systems modeled with flat dynamics include robot manipulators and classes of mobile robots, aircraft, electromechanical systems and chemical reactors. Often, establishing flatness is guided by physical insight into the system, as it has been done in this paper for MSS. \emph{Co-contraction},  the average force produced by a pair of agonist-antagonist (opposing) muscles is a fundamental indicator of the effort used to produce movement. 

A set of outputs formed by the joint coordinates and muscle co-contractions is shown below to be a flat parameterization for MSS. These systems are over-actuated, in the sense that more than one control input is collocated with each joint of the rigid body subsystem. However, there are more states than control inputs, and the muscular actuation and rigid body subsystems are dynamically coupled. As long as an invertible transformation mapping states and controls to flat outputs and their derivatives is specified, over-actuation is not a factor in establishing flatness.

Motion planning by exploiting flatness can be considered a form of inversion-based trajectory optimization, which has been studied for two decades~\cite{vanNieuwstadt98,martinCDSreport,Louembet2010}. In this context, flatness parameterizations produce the maximum possible reduction in the number of optimization variables, while direct collocation produces none~\cite{PetitIFAC2001}. 
\subsection{Overview of Sum of Squares Optimization}\label{sos}
A multivariate polynomial $p(x)$ is a sum of squares (SOS) if $p(x)=\sum_{i=1}^s h_i^2(x)$ for some polynomials $h_i(x)$, $i=1,2...s$, with $x=[x_1,x_2...x_n]$. All SOS polynomials are non-negative, but the converse does not hold. However, SOS and non-negativity are equivalent for certain important cases, namely univariate polynomials, all polynomials of degree 2, and bivariate polynomials of degree 4~\cite{Hilbert33,Reznick96}. 

The following result~~\cite{Parrilo2003} can be used to verify or enforce that a univariate polynomial is SOS: $p(x)$ of degree $2d$ is SOS if and only if there is a positive semidefinite matrix $Q$ (called the Gram matrix) such that $p(x)=z_m^TQz_m$, where $z_m(i)=t^{i-1}$, $i=1,2,..d+1$ is the vector of monomials of degree no larger than $d$. When $Q$ exists, the SOS decomposition is revealed by finding $V$ such that $Q=V^TV$ (Cholesky factorization) and observing that $p(x)=||Vz_m||^2$. 

SOS methods are attractive in optimization because they permit the reformulation or relaxation of hard, non-convex problems into much more tractable versions that may be efficiently solved through semidefinite programming (SDP)~\cite{Ahmadi2013, Parrilo2003}. Linear optimal objectives subject to SOS constraints become SDP problems for which efficient solution methods are available. Many applications to systems and control theory and optimization have been developed in the SOS framework, including Lyapunov function searches and region of attraction computations and applications to robotics~\cite{Majumdar14}.

SOS techniques are particularly well-suited to tackle motion optimization through the flatness parameterization of MSS models. Non-negativity constraints arise naturally for muscle forces, which are fundamentally tensile. Likewise, total co-contraction leads to a linear objective function and an SDP formulation.
\section{Musculoskeletal System Model}~\label{MSSmodel}
The MSS is modeled with two coupled subsystems: a serial arrangement of $N$ rigid links and a set of $m$ muscle actuators. A set of kinematic and force constraints links the subsystems. This is represented in Fig.~\ref{muscleLink}. The linkage is described by the standard robot dynamics:
\begin{equation}
M(q)\ddot{q}+C(q,\dot{q})\dot{q}+g(q)=\tau \label{robdyn}
\end{equation}
where $q$ is the $N$-vector of joint angles, $M(q)$ is the mass matrix, $C(q,\dot{q})$ is a matrix capturing centripetal and Coriolis effects and $g(q)$ is the torque due link weights. Vector $\tau$ represents the torques exerted by the muscles and constitutes the coupling from muscle dynamics. 
\begin{figure}
\centering \includegraphics[width=\linewidth]{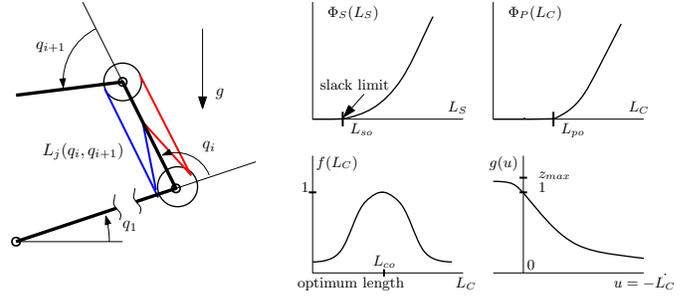} \caption{Left: Generic section of a rigid linkage driven by muscle actuators. Muscles are found in opposing pairs (agonist-antagonist), shown in blue and red. Right: General shapes of the four nonlinear functions involved in the dynamic model of the muscle actuators.}\label{muscleLink}
\end{figure}

Individual muscles are described using the Hill dynamic model~\cite{Hill38,katz1939relation,Zajac}, which considers muscles as a series-parallel arrangement of nonlinear springs and a force generator called contractile element (CE). The series elastic element (SE) represents the tendon, and it produces tension forces beyond a minimum length $L_{so}$ called the slack limit. Below the slack limit, the tendons produce zero force. That is, the SE is incapable of producing compressive forces. Tendons are attached to the linkage, producing joint torques as described in Section~\ref{kinFcoupling}.
The contractile element (CE) and the parallel elastic element (PE) have the same length, $L_C$. The PE represents the inherent elasticity of the muscle fibers, and the corresponding force is also non-negative.  
Tendon and PE forces are nonlinear functions of the respective lengths, given by
\begin{eqnarray}
 F_T &=& \Phi_S(L_S) \label{SEforce} \\
 F_P &=& \Phi_P(L_C). \label{PEforce}
\end{eqnarray}
The CE produces a tension force in proportion to an activation input $a$ and two factors that depend on the length $L_C$ and contraction rate $-\dot{L}_C$ of the CE. That is:
\begin{equation}
 F_{CE}=aF_{max}f(L_C)g(-\dot{L}_C) \label{CEforce}
\end{equation}
where $g$ is the rate dependence function with the property $g(0)=1$, $f$ is the length dependence function with the property $0<f(L_C) \leq 1$ for all $L_C$. The maximum $f=1$ is achieved at $L_C=L_{co}$, called the optimal CE length. Thus, when $a=1$, $\dot{L}_C=0$ and $L_C=L_{co}$, the CE produces its maximum isometric force, $F_{max}$. Activations are required to satisfy $a \in [0,1]$.  The general shapes of functions $f$ and $g$ are shown in Fig.~\ref{muscleLink} (right). The results of this paper do not depend on the exact shapes of these functions. For notational convenience define $u=-\dot{L}_C$.  Functions $f$ and $g$ are only assumed to satisfy the following properties:
\begin{eqnarray}
0 < f(L_C) \leq  1=f(L_{co}) & \mbox{ for all } L_C>0 \label{hillF1} \\
0 < g(u)  \leq z_{max}  & \mbox{ for all } u \label{hillG1} \\
g(0)= 1, &\;  g \mbox{ is decreasing and invertible}   \label{hillG2}
\end{eqnarray}
where parameter $z_{max}$ is known as maximum eccentric to isometric ratio, found to be near 1.5 in human muscles~\cite{Zajac}. Specific definitions of $\Phi_S$, $\Phi_P$, $f$ and $g$ are found in the literature that satisfy the above assumptions.

The activation $a$ is often modeled as the output of first-order lag dynamics with a variable time constant~\cite{Zajac}, where the input is the neural excitation $n$, regarded as the control input to be designed. The mapping from $n$ to $a$ is known as activation dynamics. Since $a$ is constrained to $[0,1]$ and the activation dynamics are modeled with a unity gain, $n$ is also constrained to this interval. In experimental biomechanics, neural inputs are often measured with electromyography sensors, whose raw signals are processed to normalize $n$ to this range.  A mathematical description of the activation dynamics is given below:
\begin{equation}
\dot{a} = \sigma(n)(n-a) \label{actdyn}
\end{equation}
where $\sigma(n)$ is a function satisfying $0< T_{min}^{-1} \leq \sigma(n) \leq T_{max}^{-1}$ for all $n$, with $\sigma(0)=T_{min}^{-1} $, $\sigma(1)=T_{max}^{-1}$ and $\sigma$ monotonically decreasing between these values. $T_{max}$ and $T_{min}$ are the maximum and minimum time constants. 
This model captures the empirically observed difference between time constants when activating ($\dot{a}>0$) and deactivating ($\dot{a}<0$) the muscle.

For the subsequent model development, it is convenient to express $u$ as a function of other variables. The total length $L$ of a muscle is the sum of the SE and CE lengths:
\begin{equation}
L=L_{S}+L_{C}. \label{L}
\end{equation}
Also, the tendon force is the sum of the CE and PE forces: 
\begin{equation}
F_T=\Phi_S(L_{S})=F_{CE}+\Phi_P(L_{C}). \label{FT}
\end{equation}
Solving for $-\dot{L}_C$ from Eq.~\ref{CEforce} and using Eqs.~\ref{L} and~\ref{FT} gives
\begin{equation}
-\dot{L}_C=u=g^{-1}(z) \label{udef}
\end{equation}
where $z$ is defined as
\begin{equation}
 z= \frac{\Phi_S(L_S)-\Phi_P(L-L_S)}{aF_{max}f(L-L_S)}. \label{zdef}
\end{equation}

\subsection{Kinematic and Force Constraints}\label{kinFcoupling}
The length of muscle $j$ is defined by Eq.~\ref{L} as $L_j=L_{Sj}+L_{Cj}$. However, each end of a muscle is attached to a link which is in motion. Therefore $L_j$ must also be a function of the link angles $q$. The characteristics of $L_j(q)$ depend on how muscle-linkage attachment is modeled. Here we adopt a simple and physiologically meaningful linear description that has been also been used in~\cite{Jagodnik}:
\begin{equation}
 L_j(q)=l_{oj}-\sum_{i=1}^N d_{ij}q_i \label{Lj}
\end{equation}
where $l_{oj}$ are the lengths at $q=0$ and $d_{ij}$ coincides with the moment arm of muscle $j$ upon joint $i$ by the principle of virtual work~\cite{Jagodnik}. This can be expressed in vector form as in Eq.~\ref{tau} below, where $\Phi_S$ denotes a vector function with components $\Phi_{Sj}(L_{S})$, $L_S$ is the vector of SE lengths and the moment arm matrix $A$ is defined by $A(i,j)=d_{ij}$, $i=1,2..N$, $j=1,2,..m$.

The kinematic and force constraints are then summarized by combining Eqs.~\ref{L} and~\ref{Lj} and including the torque relationship:
\begin{eqnarray}
L_S &=& l_o-A^Tq-L_C \label{LS} \\
\dot{L}_S &=& -A^T\dot{q}+u \label{LSdot} \\
\tau(L_S) &=& A \Phi_S(L_S) \label{tau}
\end{eqnarray}
Although this is not the only possibility, the components of $L_S$ are chosen to be state variables for the muscles, along with $a$, the vector of activations. From Eqs.~\ref{udef} and~\ref{zdef} applied to vector components and from the kinematic constraints, $u$ can be expressed as a function of the states.
\subsection{MSS Dynamics}
The MSS dynamics are represented by the following equations:
\begin{eqnarray}
M(q)\ddot{q}+C(q,\dot{q})\dot{q}+g(q) & = & A\Phi_S(L_S) \label{MSS1} \\
\dot{L}_S &=& -A^T\dot{q}+u \label{MSS2} \\
\dot{a}_j  &=& \sigma_j(n_j)(n_j-a_j) \label{MSS3} 
\end{eqnarray}
where $u_j=g_j^{-1}(z_j)$ and 
\begin{equation}
z_j= \frac{\Phi_S(L_{Sj})-\Phi_P(L_j(q)-L_{Sj})}{a_jF_{max}f(L_j(q)-L_{Sj})} \label{zjdef}
\end{equation}
for $j=1,2..m$. Note that $f,g,\Phi_{S}$ and $\Phi_{P}$ are allowed to be different for each muscle, but this has not been reflected in the notation for simplicity.
\section{Flatness Parameterization of the MSS}~\label{MSSflat}
The MSS is differentially flat, as shown by the proposed physically-motivated parameterization below. The linkage has $2N$ state variables ($N$ positions and $N$ velocities), and the muscles have $2m$ state variables ($m$ SE lengths and $m$ activations). The model has therefore $n=2(N+m)$ state variables and $m$ control inputs. Thus, $m$ flat outputs must be specified. It is assumed that $p=m-n \geq 1$, that is, there is at least one redundant muscle. 

Let $y$ be the proposed flat output vector. Define its first $N$ components as $y_i=q_i$, $i=1,2...N$ and the remaining $p$ components by $y_i=Y_i$, with $Y=E\Phi_S(L_S)$, where $E$ is a $p$-by-$m$ matrix defining $p$ flat outputs as the averages of the agonist-antagonist muscle tendon forces. If $p>m/2$, the extra rows of $E$ are defined as arbitrary linear combinations of $\Phi_{Sj}(L_{S})$, under the restrictions listed below. Define matrix $C$ by:
\begin{equation}
 C =   \left [ \begin{array}{c}
                 A \\ 
                 E
                \end{array} \right ]. \label{Cdef}
\end{equation}

To establish flatness and ensure the feasibility of subsequent optimization problems, several assumptions are made. Let $\bf{1}_m$ denote a vector with all $m$ components equal to one.
\begin{A}\label{A1}
The moment arm and flat ouput definition matrices $A$ and $E$ satisfy the following conditions:
\begin{enumerate}
 \item $C$ is a non-singular $m$-by-$m$ matrix.
 \item The row sums of $E$ equal 1. That is, $\sum_{j=1}^m E_{ij}=1$ for $i=1,2,..p$.
 \item Let $C^{-1}$ be partitioned as $C^{-1}=[C_\tau | C_Y]$, where $C_Y$ has $p$ columns. Let $\sigma_\tau$ denote the vector of row sums of $A$. The following must hold:
 \[{\bf{1}}_m-C_\tau \sigma_\tau \succ 0. \]
\end{enumerate}
\end{A}
The above conditions are immediately verified for an agonist-antagonist moment arm matrix $A$ with symmetric muscle attachments ($\sigma_\tau=0$) and $E$ based on co-contractions and tendon forces. The last assumption is placed to guarantee feasibility of the reduced linear program. 

\begin{Prop}
Consider the transformation $\Psi$ defining the flat outputs $y$ from system states:
\begin{eqnarray}
y=\Psi(q,\dot{q},L_S) &=& \left [ \begin{array}{c} q \\ Y \end{array} \right ] \label{ydef1} \\
 Y &=&  E \Phi_S(L_S) \label{ydef2}
\end{eqnarray}
and suppose Assumption~\ref{A1} holds. Then $\Psi$ is invertible, providing a flatness parameterization of the MSS.
\end{Prop}
\begin{proof}
To construct the inverse, first note that the joint torque $\tau$ can be expressed as a function of the first $N$ flat outputs and derivatives up to second order by using Eq.~\ref{robdyn}. With a slight abuse of notation, denote this computed torque by $\tau(y,\dot{y},\ddot{y})$. Thus 
\begin{equation}
\left [ \begin{array}{l} 
         \tau(y,\dot{y},\ddot{y}) \\
         Y
        \end{array} \right ] =
        \left [ \begin{array}{c}
                 A \\ 
                 E
                \end{array} \right ] \Phi_S(L_S) = C \Phi_S(L_S). \label{fp1}
\end{equation}
Assumption~\ref{A1} requires that $C$ be invertible, then states $L_S$ can be obtained from the flat outputs provided $L_S$ remains outside the slack limit, where functions $\Phi_{Sj}$ are themselves invertible.  To show that states $a$ are parameterized by $y$ and its derivatives, let $c_j$ denote the $j$-th row of $C^{-1}$. Then inverting Eq.~\ref{fp1} and using the derivative of inverse formula it can be checked that
\begin{equation}
\dot{L}_{Sj}=\frac{[\dot{\tau}^T\; \; \dot{Y}^T]c_j^T}{\Phi'_{Sj}(L_{Sj})} \label{LSdotflat}
\end{equation}
where $\dot{\tau}$ is a function of the first $N$ flat outputs and their derivatives. Then $u_j$ can be found from Eq.~\ref{LSdot} and $L_j$ can be found from Eq.~\ref{Lj}. Therefore $z_j$ can be calculated as $g_j(u_j)$, and Eq.~\ref{zdef} gives states $a_j$ as a function of flat outputs and derivatives. In turn, differentiation gives $\dot{a}_j$, which yields the control inputs $n_j$ through solution of Eq.~\ref{MSS3}.  The formalism defining the endogenous transformation $\Psi$ involves an infinite prolongation of its arguments to accommodate as many derivatives of the control and the flat outputs as required~\cite{martinCDSreport}. In our case, $\Psi$ does not depend on $n$, and $\Psi^{-1}$ requires up to the third derivative of $Y$.  
\end{proof}
In the optimization proposed below, polynomials are used for $q(t)$, which simplifies the computation of the first $N$ flat outputs and all their derivatives. The torque $\tau(y,\dot{y},\ddot{y})$ is directly evaluated from Eq.~\ref{robdyn}. Likewise, $L$, $L_S$, $\dot{L}_{S}$, $u$, $z$ and $a$ constitute direct function evaluations based on flat outputs and derivatives, which are all polynomials. The computation of $n$, however, requires numerical differentiation and solution of $m$ decoupled nonlinear equations. 
\section{Optimization Problems}~\label{opt}
The objective of optimization is to provide an efficient method to specify joint trajectories $q(t)$ in an interval $[0,T]$ that satisfy given boundary values, along with corresponding solutions for the muscle states $L_S(t)$ and $a(t)$ and control inputs $n(t)$ that satisfy a set of constraints and arise from a minimum-effort criterion.

For joint trajectories, optimality is used to replicate qualities of natural motion such as smoothness and absence of large overshoot. For muscle states and controls, optimality criteria must reflect the often-invoked principle~\cite{acker10} that motion must be completed with the smallest effort. Due to the high redundancy found in MSS, it is possible to achieve the same joint trajectories with varying levels of muscle effort, and optimality is used to resolve control redundancy. 

Evidently, the trajectories resulting from optimization are dictated by the specific cost functions used to penalize unnatural motions and large muscle efforts, and a wide variation of cost functions and their parameters exist in the literature. Constraints have been formulated more consistently, to guarantee that activations and neural inputs remain in the interval $[0,1]$ and that tendon force solutions are non-negative (slack prevention).

In this work, a two-stage process is adopted, whereby joint trajectories that meet prescribed boundary values are selected first. The optimization criterion involved in this selection, if any, is relative to $q(t)$ alone. For instance, an integral-square error (ISE) criterion may used to transfer the joints between the specified boundary conditions. For tracking control, the trajectory $q(t)$ is preset and optimization does not apply. In this paper we consider that a function $q(t)$ is fit to the applicable boundary conditions. Polynomials are straightforward for this purpose, since their coefficients can be obtained by solving a system of linear equations. Regardless of the method, the selection of $q(t)$ will be referred to as Primary Motion Planning Problem (PMPP). The joint trajectories selected in the first stage are then used as data in the secondary optimization problem (SOP), which is based on a minimum-effort criterion.

The proposed two-stage optimization approach may be used in two settings: (A) open-loop solutions and (B) closed-loop receding-horizon (RH) implementations, for instance  model predictive control (MPC)~\cite{RawlingsBook,GruenePannek}, where optimization is solved repeatedly with initial conditions obtained from state feedback. In (A), the goal is to determine feasible neural input trajectories that optimally transfer the system between two points in a given time horizon. The solution is valid in the prescribed horizon, subject to an exact correspondence with the model and in the absence of external influences such as disturbance forces. In (B), the feedback introduced through state measurements and recursive optimization is intended to bring tolerance for model errors, disturbances and changes in reference commands. 

\noindent \emph{Secondary Optimization:} The $q(t)$ resulting from the PMPP defines the first $N$ flat outputs. The remaining $p$ are obtained under a minimum-effort criterion and a set of constraints, divided in two groups: boundary equality constraints and muscle slack inequality constraints. The specific set of boundary constraints is again determined by problem setting. Depending on whether initial condition matching and/or terminal equilibrium is sought, constraints may exist on $Y$ and its derivatives. 

The flatness parameterization can be used to show that a point constraint on the MSS state introduces corresponding constraints on $Y$ and $\dot{Y}$, as well as the acceleration and jerk at the same point. Moreover, it can be shown that equilibrium conditions require that that $Y$ and its first three derivatives vanish. Thus, the smallest degree to be considered for the polynomials in $Y(t)$ depends on the equality constraints that are included.  

Muscle slack constraints are captured by a non-negativity requirement on pertinent polynomials. Tendon forces are required to remain at or above a set of force reserves given by vector $\underline{F}_{T}$; that is, $\Phi_S(L_S) \succcurlyeq \underline{F}_T$ is enforced, where the symbol $\succcurlyeq$ is used to indicate component-wise inequalities. From Eq.~\ref{fp1}, tendon forces can be written as
\[\Phi_S(L_S)=C^{-1}\left [ \begin{array}{l} 
         \tau(y,\dot{y},\ddot{y}) \\ 
         Y
        \end{array} \right ]= C_\tau \tau+C_Y Y. \]
However, $\tau(y,\dot{y},\ddot{y})$ is non-polynomial for general linkage model matrices $M(q),C(q)$ and $g(q)$, and it is not possible to express the constraint as a non-negative polynomial inequality. Bounds on $\tau$ are used to overcome this difficulty. Since $y_i$ have been calculated for $i=1,2...N$ in the PMPP, $\tau$ can be numerically evaluated in the interval $[0,\;T]$ and bounds $\underline{\tau}_i$, $\bar{\tau}_i$ extracted. Such function evaluation can be performed at high speed and its outputs are also useful in subsequent calculations. As discussed in Section~\ref{intro}, there is no built-in conservativeness in requiring SOS properties for univariate polynomials of even degree as a means to enforce their non-negativity. 
The SOS inequality then takes the form $ C_Y Y(t) \succcurlyeq B$, where each entry of $B$ can be found from the linear program: 
\[ B_i=\underline{F}_{T,i}-\mathop{\mbox{min}}_{\underline{\tau} \leq \tau \leq \bar{\tau}} \{ c_{i} \tau \} \]
where $c_i$ is the $i$-th row of $C_\tau$. The muscle effort measure to be minimized in this paper is given by the integral of the co-contractions, which is a linear objective in the space of polynomial coefficients. Note that the non-negativity constraint on $Y$ eliminates the need for square or absolute value measures in the cost function. 

The SOP is now formulated as follows:
\begin{equation}
\begin{aligned}
& \underset{Y}{\text{minimize}} \;\; \int_0^T \sum_{i=1}^p Y_i(t) dt \text{  \emph{subject to}} \\
& C_Y Y(t) \succcurlyeq B,\;Y(t) \succcurlyeq 0 \\
& \mbox{\emph{applicable terminal and initial constraints on} } Y, \dot{Y}, \ddot{Y} \\
\end{aligned} \label{SOP01}
\end{equation}
 
As discussed in Section~\ref{intro}, optimization of an objective which is linear in the coefficients of the SOS polynomials subject to equality and SOS inequality constraints is equivalent to a semidefinite program. A variety of efficient solvers are available that perform the necessary symbolic processing to encode a given problem as a semidefinite program. In this paper, the sum of squares optimization toolbox for Matlab SOSTOOLS~\cite{sostools} is used in the examples.
\subsection{Recursive Feasibility of the SOP}
Recursive feasibility is of central importance to RH implementations such as MPC. An RH implementation involves solving the optimal control problem at the initial time $t_0$, followed by application of the optimal control restricted to the interval $[t_0,\;t_0+\delta]$. At time $t_0+\delta$, the closed-loop response of the plant is used to update the initial constraints and the optimization problem is solved again. The process is repeated indefinitely, establishing a feedback process based on open-loop optimal solutions. 

Let the constraints for $Y$ at time $t_0$ be described by $V_{Y,t_0}=v(x(t_0))$ where $x$ is the state of the MSS dynamics and $V_{Y,t_0}$ is a vector containing $Y$ and a finite number of its derivatives evaluated at $t_0$, as defined by a specific SOP formulation. Function $v$ uses the flatness parameterization to define constraint values from plant state data. We assume that terminal equilibrium constraints are used, i.e., $Y(t_0+T)=\bar{Y}$ and $Y^{(d)}(T)=0$ for derivative orders $d=1,2$ or higher, as specified by the problem.

Suppose a solution $Y^*(t)$ , $t \in [t_0,\;t_0+T]$ is found, leading to corresponding optimal state and control trajectories $x^*(t)$ and $n^*(t)$. Upon application of the restriction $n_\delta(t)=n^*(t),\;t \in [t_0,t_0+\delta]$, the response of the MSS at the end of this interval is $x_{cl}(t_0+\delta)$, with corresponding $Y_{cl}(t_0+\delta)$. 

With no model uncertainties or disturbances, $x_{cl}(t_0+\delta)=x^*(t_0+\delta)$ and $Y_{cl}(t_0+\delta)=Y^*(t_0+\delta)$, thus the new instance of the SOP at $t=t_0+\delta$ involves initial constraints $V_{Y,t_0+\delta}=v(x^*(t_0+\delta))$. The aim of this section is to prove that if the SOP is feasible at $t_0$, it will remain so at times $t_0+k\delta$, $k=1,2,...$. The result below requires a strict form of feasibility. 

\begin{D}\label{sfeas}
The SOP is said to be \emph{strictly feasible at time $s$} if there is a polynomial vector $Y^*$ and $\epsilon_0>0,\; \epsilon_1 >0$
such that $V_{Y^*,s}=v(x(s))$, $C_Y Y^*-B \succcurlyeq \ \epsilon_0 \bf{1}_m$ and $Y^* \succcurlyeq \epsilon_1 \bf{1}_p$.
\end{D}

\begin{T}
Suppose there is a solution $Y^*_k(t)$ to SOP (~\ref{SOP01}) which is strictly feasible at time $t_0+k\delta$ for some $k \in \mathbb{N}$, $k \geq 1$ and $0 < \delta<T$. Let the corresponding neural inputs be $n^*(t),\; t \in [t_0+k\delta,\;t_0+k\delta+T]$.  Suppose that the restriction $n_\delta(t)=n^*(t)$ for $t \in [t_0+k\delta,\;t_0+(k+1)\delta]$ is applied to the MSS system, resulting in trajectories $x_{cl}(t)$ and $Y_{cl}(t)$ defined on the same interval. Then the SOP is strictly feasible at time $t_0+(k+1)\delta$. 
\end{T}

\begin{proof}
Construct a function $Y_c(t)$ defined in $[t_0+(k+1)\delta, \; t_0+(k+1)\delta+T]$ as below
\begin{equation}
Y_c(t)= \begin{cases}
Y^*_k(t) & \mbox{if } (k+1)\delta \leq t-t_0 < k\delta+T \\
\bar{Y} & \mbox{if } k\delta+T \leq t-t_0 \leq (k+1)\delta+T. \end{cases} \label{Ycont}
\end{equation}
Assuming there are no model errors, $Y_{cl}(t_0+(k+1)\delta)=Y^*_k(t_0+(k+1)\delta)=Y_c(t_0+(k+1)\delta)$, therefore $V_{Y_c,t_0+(k+1)\delta}=v(x_{cl}(t_0+(k+1)\delta)$, showing that $Y_c$ meets the initial constraints. Since $Y_c=\bar{Y}$ is constant in $[t_0+k\delta+T,\;t_0+(k+1)\delta+T]$, it also meets the terminal equilibrium constraints. Moreover, $Y_c$ clearly satisfies the strict feasibility inequalities of Definition~\ref{sfeas}.  However, $Y_c$ is continuous but not polynomial.

An extension of the Weierstrass theorem can be invoked~\cite{Peet07} to find a polynomial vector $\hat{Y}$ that approximates $Y_c$ within any prescribed error bound and also satisfies the initial and terminal constraints. Note that an evaluation mapping $L(Y_j)=Y_j^{(d)}(s)$ for any derivative order $d$ and time $s$ applied to a continuous function $Y_j$ constitutes a bounded linear operator $L: \mathcal{C}^d[t_0+(k+1)\delta,\;t_0+(k+1)\delta+T] \mapsto \mathbb{R}$. This permits the application of Corollary 3 of~\cite{Peet07} to this case: for any $\gamma>0$, there is a polynomial vector $\hat{Y}$ such that 
\begin{equation}
||\hat{Y}-Y_c||_\infty \leq \gamma \label{ineqGamma}
\end{equation}
with $L_i(\hat{Y}_j)=L_i(Y_{cj})$, for $j=1,2..p$ and $i=1,2,..r$, where $r$ is the number of equality constraints in the SOP (dimension of $V_Y$). 

Then $V_{\hat{Y},t_0+(k+1)\delta}=V_{Y_c,t_0+(k+1)\delta}=V_{Y_{cl},t_0+(k+1)\delta}=v(x_{cl}(t_0+(k+1)\delta))$, thus $\hat{Y}$ meets initial constraints. Likewise, it follows that $Y_c$ meets the terminal constraints. Now, since the strict feasibility inequalities for $Y^*_k$ carry over to $Y_c$:
\[C_Y\hat{Y}=C_YY_c+C_Y(\hat{Y}-Y_c) \succcurlyeq   B+\epsilon_0 {\bf{1}}_m - ||C_Y(Y_c-\hat{Y})||_\infty \bf{1}_m \]
and the submultiplicative property together with inequality (~\ref{ineqGamma}) yield
\begin{equation}
C_Y\hat{Y} \succcurlyeq  B+(\epsilon_0-\gamma ||C_Y||_\infty) \bf{1}_m. \label{strictineq0}
\end{equation}
Likewise, the above argument holds with $B=0$, $C_Y=I$ and $\epsilon_1$ to give
\begin{equation}
\hat{Y} \succcurlyeq  (\epsilon_1-\gamma) \bf{1}_p \label{strictineq1}
\end{equation}
Taking $\gamma=$ min$(\frac{\epsilon_0}{2},\; \frac{\epsilon_1}{2||C_Y||_\infty})$ shows 
\begin{equation}
C_Y\hat{Y}  \succcurlyeq  B+ \frac{\epsilon_0}{2} {\bf{1}_m} \;\mbox{and } \hat{Y}  \succcurlyeq  \frac{\epsilon_1}{2} \bf{1}_p.
\end{equation}
Therefore the SOP is strictly feasible at $t_0+(k+1)\delta$. The above shows that strict feasibility for $k$ implies the same for $k+1$, and the case $k=0$ holds by assumption. Inductively, the SOP is strictly feasible for $k=0,1,2,..$. 
\end{proof}
\emph{Remark:} The above guarantees the existence of $\hat{Y}$ without indicating the required degree. In computations, the monomial basis $z$ is chosen to allow for a sufficiently large degree. 
\section{Linear Programming Solution}
Suppose that initial states $L_S(0)$ and $a(0)$ are not specified, but rather found as a consequence of optimization. Per the model, this results in no constraints placed on the initial values or rates of $Y$, and the SOP is formulated as in (\ref{SOP01}) without the additional equality constraints at the initial time. With $\dot{Y}(T)=0$ as a constraint (as would be implied by a terminal equilibrium constraint), the following result states that the optimal polynomial vector $Y$ is constant, determined by a feasible linear program.   

\begin{T}
Suppose the SOP is formulated with terminal constraint $\dot{Y}(T)=0$ and no initial constraints for $Y$ or its derivatives. Then the optimal polynomial vector is a constant given by the solution of a linear program which is feasible under the assumptions.
\end{T}
\begin{proof}
Suppose $Y^*$ is an optimal solution, corresponding to optimal value $J(Y^*)$. Decompose it as $Y^*=\alpha_0^*+\delta$, where constant $\alpha_0^*$ is the mean value of $Y^*$ in $[0,T]$. By linearity of the objective function it follows that
\begin{equation}
 J(Y^*)=J(\alpha_0^*)+J(\delta). \label{costDecomp}
\end{equation}
But $J(\delta)=0$, since the above decomposition implies that $\delta$ has zero mean in $[0,T]$. Therefore $\alpha_0^*$ attains the same cost as $Y^*$ and constitutes an optimal solution as long as it is feasible. Feasibility of $\alpha_0^*$ indeed follows compatibility with the terminal constraint and by averaging the inequality constraints, using linearity and $\frac{1}{T}\int_0^T\delta(t)dt=0$. \\
With constant $Y$, the SLP reduces to \\
\emph{Secondary Linear Program} (SLP): 
\begin{equation}
\begin{aligned}
& \underset{}{\text{minimize}} \;\;  \sum_{i=1}^p \alpha_{i0} \text{ subject to} \\
& C_Y \alpha_0 \succcurlyeq B,\;\;\alpha_0 \succcurlyeq 0
\end{aligned} \label{SLP}
\end{equation}
The SLP is always feasible under Assumption~\ref{A1}, as shown in Appendix~\ref{app1}. 
\end{proof}
\section{Simulation Examples}~\label{simex}
The simulation examples use an MSS model of an arm, which has been previously considered ~\cite{Jagodnik}. The MSS has two degrees of freedom and six muscles, that is, $n=2, m=6, p=4$. The muscles are organized in three agonist-antagonist pairs, where one pair operates across two joints, as seen in Fig.~\ref{mdldwg}.
\begin{figure}
\centering \includegraphics[width=\linewidth]{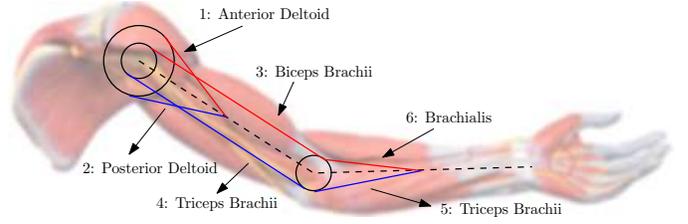} \caption{Two degree-of-freedom, six-muscle arm model used in the simulation example} \label{mdldwg}
\end{figure}
Listings of the mass, Coriolis and gravity vectors for the linkage and the corresponding parameters are available from~\cite{Jagodnik} or the authors.

The force capacities and kinematic data for the muscles are the same as in~\cite{Jagodnik} and contained in the simulation code associated with this paper~\cite{flatSOScode}. The first 2 flat outputs are $y_1=q_1\;y_2=q_2$. With six muscles, three pairwise co-contractions are defined: $Y_i=\frac{1}{2}(\Phi_S(LS_{2i-1})+\Phi_S(LS_{2i}))$, $i=1,2,3$. One additional flat output must be formed so that $C$ is full-rank. In this example, $y_6=\Phi_{S}(LS_5)$ is defined arbitrarily. It should be noted that unlike the example of~\cite{Jagodnik}, we assume that the arm moves in a vertical plane and is thus subjected to gravity torque.  All simulations were carried out with Matlab 9.5 running on an Intel Core i5, 7th generation processor. 

The first simulation involves prescribed initial conditions and terminal equilibrium. The objective is to transfer the arm from a nearly-horizontal position $q_0=[0\;10^\circ]^T$ to straight-up equilibrium. The remaining initial conditions were chosen as $\dot{q}_0=[0.001\;0.002]^T$ rad/s, $a_0=[0.135\;0.040\;0.404\;0.556\;0.414\;0.068]^T$ and $LS_0=[0.055\;0.054\;0.235\;0.192\;0.194\;0.018]^T$. 
These initial conditions and terminal equilibrium must be matched by the optimizer, requiring a full SOS polynomial solution for the optimal co-contractions $Y$. Initial condition matching places constraints on the initial joint acceleration and jerk. These values were calculated and used to construct identical polynomials for $q_1$ and $q_2$ that also meet the equilibrium constraints at the terminal time, chosen as $T=3$ s. An eight-order polynomial was required and obtained without optimization. 

A grid of 31 points was established for the calculation of torque bounds with Eq.~\ref{MSS1} and for numerical differentiation. The tendon force reserves were set to $\underline{F}_{T,i}=10$ N for all muscles.  The problem was encoded and solved using SOSTOOLS, set to use the SeDuMi semidefinite programming solver~\cite{S98guide}.
Since there are four equality constraints for $Y$, fourth-order polynomials were selected, leaving one degree of freedom per polynomial to optimize. The solution was found in 0.32 CPU seconds as reported by SOSTOOLS. For verification the optimal open-loop neural inputs were applied to a forward integration of the MSS dynamics. The results are shown in Fig.~\ref{exOLAct}. 
\begin{figure}
\centering \includegraphics[width=\linewidth]{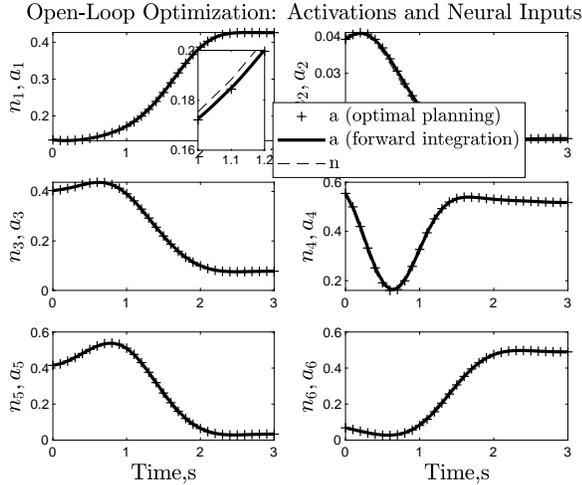} \caption{Activation histories in the open-loop simulation, comparing the optimally-planned and forward-integrated trajectories. The optimal neural inputs are also shown, displaying the small lag between $n$ and $a$ due to activation dynamics.} 
\label{exOLAct}
\end{figure}

The second  simulation illustrates the application of the proposed solution method for RH control. A closed-loop controller is implemented to transfer the position of the arm between two equilibrium positions. Cubic polynomials $q_1(t)$ and $q_2(t)$ are first selected that meet the pertinent initial and final boundary values. The resulting trajectories define the torque histories and the required bounds. 

The boundary conditions and control objectives in this simulation match those of Jagodnik~\cite{Jagodnik}. The initial position is $q_0=[20^\circ\;20^\circ]^T$ and the arm is to be transferred to $\bar{q}=[80^\circ\;80^\circ]^T$. The problem is solved with a prediction horizon of $T=0.5$ s discretized with 11 points ($\delta=0.05$ s). Optimization is solved at each time of the form $t=k\delta$, with $k=0,1,..$,  followed by application of the neural control inputs in the interval $[k\delta,\;(k+1)\delta]$. The position and velocity response to these inputs are used as initial conditions for the next optimization. No additional initial constraints are placed, and equilibrium conditions are requested at the terminal time of each optimal trajectory prediction. Therefore the problem reduces to linear programming, as discussed above. The tendon force reserves were all set to 1.

Using Matlab's {\tt linprog}, each solution is completed in as little as 0.015 seconds, with additional time required for post-optimization calculations and plant update integrations, perfomed with Matlab's {\tt ode23}. The total simulation time for each time step was 0.06 seconds, which is approximately real-time considering the value of $\delta$. Coarser simulations will, of course, run much faster than real-time. The simulation code is available~\cite{flatSOScode}. The results are shown in Figs.~\ref{exCLJoints} and~\ref{exCLn}. The results appear smoother than those generated in~\cite{Jagodnik} and completed within the same settling time. 
\begin{figure}
\centering \includegraphics[width=\linewidth,height=0.6\linewidth]{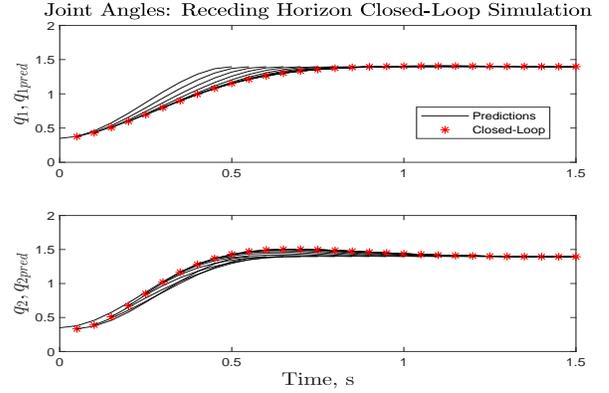} \caption{Predicted and closed-loop joint angles in the receding-horizon implementation.} \label{exCLJoints}
\end{figure}
\begin{figure}
\centering \includegraphics[width=\linewidth, height=0.6\linewidth]{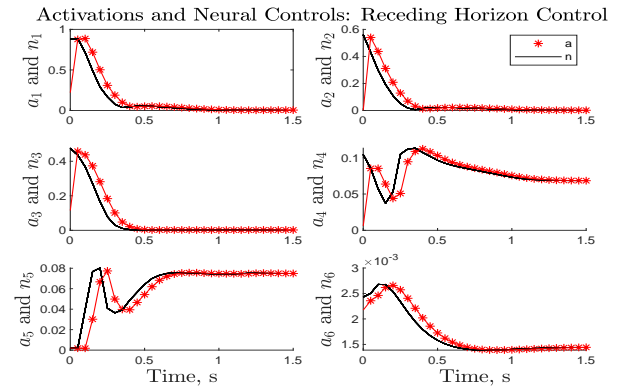} \caption{Activations and neural input histories in the closed-loop receding-horizon implementation.} \label{exCLn}
\end{figure} 
The difference between the predicted and actual closed-loop trajectories, more noticeable in joint 1, is expected with receding-horizon implementations. The associated suboptimality is well-understood and even quantifiable~\cite{GruenePannek}. In general, the influence of the prediction horizon on the closed-loop performance is difficult to characterize. 
\section{Concluding Remarks}
In comparison with collocation methods widely used in human motion control studies, the proposed approach does not include direct constraints on activation states or neural control inputs. In exchange, the proposed method can generate solutions in real time or faster and lends itself to analysis resulting in a recursive feasibility guarantee and simplification to a linear program. Further study is warranted concerning the use of the approach in fast optimal solvers for model predictive control. Our proof of recursive feasibility hinges on an accurate model. The effects of uncertainty deserve additional considerations. 
\section*{Acknowledgments}
The authors would like to acknowledge the support of the National Science Foundation, Cyber-Physical Systems Program through grant \# 1544702  Also, we acknowledge the reviewers for their valuable suggestions to streamline the main results.
\appendices
\section{Proof of Feasibility of the SLP}\label{app1}
Claim: For $C$ satisfying Assumption~\ref{A1}, the 
inequality $-C_Y\alpha_0+B \preccurlyeq 0$ admits non-negative solutions $\alpha_0$. 

To prove this, it is first shown that the row sums of $C_Y$ are positive, that is, $C_Y \mathbf{1}_p \succ 0$. From the assumption, it follows that
\[C\mathbf{1}_{m}=[\sigma_\tau^T\;|\;\mathbf{1}_p^T]^T\]
Since $C$ is invertible this means
\[C^{-1}[\sigma_\tau^T\;|\;\mathbf{1}_p^T]^T=C_\tau \sigma_\tau+C_Y \mathbf{1}_p=\mathbf{1}_{m}\]
thus
\[C_Y\mathbf{1}_p=\mathbf{1}_{m}-C_\tau \sigma_\tau \succ 0\]
Next, consider the SLP constraints
\begin{equation}
C_Y\alpha_0 \succcurlyeq B,\;\alpha_0 \succcurlyeq 0 \nonumber
\end{equation}
A feasible solution can be constructed by taking $\alpha_0=\gamma \sigma$
where $\sigma=C_Y\mathbf{1}_p \succ 0$ and the scalar $\gamma$ is chosen as follows
\[\gamma=\mbox{ max } \{ \mbox{ max } \{b_i/\sigma_{i} \},0 \}\]
Clearly $\alpha_0 \succcurlyeq 0$ and it can be directly verified that the first SLP constraint is satisfied. 
\bibliographystyle{IEEEtran}

\bibliography{AdvExercise,ctrlTHcourse,MPCcourse,SOS}

\end{document}